\documentclass[journal]{IEEEtran}

\usepackage{latexsym,amsmath,amssymb}

\def \epf {\hfill $\Box$\\ \par}

\def\F {{\mathbb{F}}}

\def\bc{{\bf c}}
\def\be{{\bf e}}

\def\bu{{\bf u}}
\def\bv{{\bf v}}

\def\bo{{\bf 0}}

\def\bx{{\bf x}}
\def\bz{{\bf z}}

\def\bs{{\bf s}}
\def\be{{\bf e}}

\def\CC{\mathbb{C}}

\def \mC {\mathcal{C}}

\def \mG {\mathcal{G}}
\def \mC {\mathcal{C}}

\def \mF {\mathcal{F}}

\def \mS {\mathcal{S}}

\def \Pr {{\rm Pr}}
\def\wt{{\rm wt}}

\newcommand{\Ga}{\alpha}

\newcommand{\Gd}{\delta}

\newcommand{\Gl}{\lambda}    
     
\newcommand{\Gs}{\sigma}     \newcommand{\GS}{\Sigma}

\newtheorem{thm}{Theorem}[section]
\newtheorem{defn}[thm]{Definition}

\newtheorem{lemma}[thm]{Lemma}

\numberwithin{equation}{section}

\begin{document}


\title{On the List-Decodability of Random Self-Orthogonal Codes}

\author{Lingfei Jin, Chaoping Xing and Xiande Zhang
\thanks{L. Jin is with Shanghai Key Laboratory of Intelligent Information Processing, School of Computer Science, Fudan University, Shanghai 200433, China. {\it email:} {lfjin@fudan.edu.cn}. }
\thanks{C. Xing is with Division of
Mathematical Sciences, School of Physical \& Mathematical Sciences,
Nanyang Technological University, Singapore 637371, Republic of
Singapore. {\it email:} {xingcp@ntu.edu.sg}}

\thanks{X. Zhang is with Division of
Mathematical Sciences, School of Physical \& Mathematical Sciences,
Nanyang Technological University, Singapore 637371, Republic of
Singapore. She is also with School of Mathematical Sciences, University of Science and Technology of China, Hefei, Anhui, 230026 China. {\it email:} {xiandezhang@ntu.edu.sg}}
\thanks{The work of C. Xing was supported by the Singapore Ministry of Education under the Tier 1 Grant RG20/13.
X. Zhang is partially supported by NSFC under gratnt 11301503.}}



\maketitle

\begin{abstract} In 2011, Guruswami-H{\aa}stad-Kopparty \cite{Gru} showed that the list-decodability of random linear codes is  as good as that of general random codes. In the present paper, we further strengthen the result by showing that the list-decodability of random {\it Euclidean self-orthogonal} codes is  as good as that of general random codes as well, i.e., achieves the classical Gilbert-Varshamov bound. Specifically,
we show that, for any fixed finite field $\F_q$, error fraction
$\delta\in (0,1-1/q)$ satisfying $1-H_q(\delta)\le \frac12$ and
small  $\epsilon>0$,  with high probability a random Euclidean self-orthogonal
code over $\F_q$ of rate $1-H_q(\delta)-\epsilon$ is $(\delta,
O(1/\epsilon))$-list-decodable. This generalizes the result of
linear codes to Euclidean self-orthogonal codes. In addition, we extend the result to list decoding {\it symplectic dual-containing} codes by showing that the list-decodability of random symplectic dual-containing codes achieves the quantum Gilbert-Varshamov bound as well. This implies that list-decodability of quantum stabilizer codes can achieve the quantum Gilbert-Varshamov bound.
 The counting
argument on self-orthogonal codes is an important ingredient to
prove our result.
\end{abstract}


\section{Introduction}
The notion of list decoding was introduced independently by Elias
and Wozencraft \cite{El, El91, Wo}. Instead of insisting on a unique
output of codeword, in the list decoding model the decoder allows
to  output a list of possible codewords which includes  the actual
transmitted codeword. Compared with the classical unique decoding
model, the model of list decoding  allows  larger number of
corrupted errors. A fundamental problem in coding theory is the
trade-off between the information rate and the fraction of errors
that can be corrected.  For list decoding, we have another important
parameter, i.e., the largest list size of the decoder's output. We
hope the list size to be small.

From the algorithm point of view, a good list decoding algorithm
should have  polynomial time, which means that the list size should
be at most polynomial in the block length of the code. Researchers
have been devoted to finding good list decodable codes  with
efficient list-decoding algorithms due to the wide applications to
complexity theory and more general for computer science
\cite{Gru01,Su,Vad11}, and communications \cite{El91}. The fraction
of errors $\delta $ close to $1-1/q$ is more interesting for
complexity theory,  while it is more attractive for $\delta $ close
to $0$ for communication side. Thus, it is meaningful to consider
the full range of possibilities for $\delta$.

\subsection{The Gilbert-Varshamov bound}
Before starting our paper, we first introduce the Gilbert-Varshamov bound in coding theory that plays a central role in this paper.

For an integer $q\ge 2$, we define the $q$-ary
entropy function by
$H_q(x)=x\log_q(q-1)-x\log_q x-(1-x)\log_q (1-x)$. Then it is easy to verify the identity $H_{q^2}(x)=\frac12H_q(x)+\frac12x\log_q(q+1)$.  It has been proved that, with high probability, a random $q$-ary classical block code (and a random $q$-ary classical linear block code, respectively) of sufficiently large length with rate $R$ and relative  Hamming  minimum distance $\delta$ satisfies the following $q$-ary classical Gilbert-Varshamov bound \cite{Var}
\begin{equation}\label{eq:1.1} R\ge 1-H_q(\Gd).\end{equation}
 Similarly, with high probability, a random $q$-ary quantum code  of sufficiently large length with rate $R$ and relative  symplectic minimum distance $\delta$ satisfies the following $q$-ary quantum Gilbert-Varshamov bound \cite{Ash Kni}
\begin{equation}\label{eq:1.1} R\ge 1-H_q(\Gd)-\Gd\log_q(q+1).\end{equation}

\subsection{Status of list decoding random codes}

It is well known that the list-decodability of classical block codes is upper bounded by the  classical Gilbert-Varshamov bound (see \cite{Gru01}), i.e., the  tolerance error rate $\delta\leq H_q^{-1}(1-R)$. On the other hand, it
was shown in \cite{El91} that for  a random code  with rate $R\leq
1-H_q(\Gd)-1/L$, it is $(\delta, L)$-list-decodable with probability
at least $1-q^{-\Omega(n)}$. However, it is not obvious to
generalize this result to linear codes.

Zyablov-Pinsker \cite{ZP} established an optimal tradeoff between
the rate $R$ and the fraction of errors $\delta$ for binary linear
codes. The results in \cite{ZP} can be easily generalized to $q$-ary
codes which shows that the minimum list size of a linear code with
rate $1-H_q(\delta)-\epsilon$ is bounded by $\exp
(O_q(1/\epsilon))$. But this bound is exponentially worse than the
bound $O(1/\epsilon)$ for arbitrary codes.

In \cite{Gru02}, Guruswami-H{\aa}stad-Sudan-Zuckerman showed
existence of $(\delta,1/\varepsilon)$-list-decodable linear codes of
rate at least $1-H_2(\delta)-\varepsilon$ for binary codes. Although
the extension of  the results in \cite{Gru02} to larger alphabets is
not easy, Guruswami-H{\aa}stad-Kopparty \cite{Gru} finally managed
to show that a list size of $O_q(1/\epsilon)$ suffices to have rate
within $\varepsilon$ of the information-theoretically optimal rate
of $1 - H_q(\Gd)$. This means that the list-decodability of random
linear codes is as good as that of general codes. In the latest
development, M. Wootters \cite{W13} improved the constant in the
list size $O_q(1/\epsilon)$ for random linear codes when the
decoding radius $\Gd$ is close to $1-1/q$.
\subsection{Motivation}
It is well known that (symplectic) self-orthogonal codes form a useful and
important class  of linear codes which have found wide applications
in communications \cite{HST99,NRS05}, multiplicative secret sharing
\cite{CCGHV07} and quantum codes \cite{CRSS98}, etc.. It is natural
to ask the question about how well one can list decode a random (symplectic)
self-orthogonal code or dual-containing code (a symplectic dual-containing code is a subspace of $\F_q^{2n}$ that contains its dual under the symplectic inner product).

Euclidean self-orthogonal codes are extensively used for construction of linear multiplicative secret sharing
\cite{CCGHV07}. In the event that some dishonest players provide fault shares, we have to carry on error correction to recover the secret. In this scenario, one has to consider decoding of Euclidean self-orthogonal codes.

In quantum coding theory, decoding of a quantum stabilizer code $Q$ obtained from a classical self-orthogonal code $C$ can be reduced to  decoding of the symplectic dual code $C^{\perp_S}$ (see Section III.C for details). Therefore, list decoding of dual-containing codes with symplectic inner product plays an important role on quantum decoding.

\subsection{Our work and techniques}

In this work, we focus on list decoding of Euclidean self-orthogonal and symplectic dual-containing codes. Surprisingly, our results show that
the list-decodability of random Euclidean self-orthogonal codes and symplectic dual-containing codes  are  as good as
that of general random codes and random linear codes, namely, the list-decodability of random Euclidean self-orthogonal codes and symplectic dual-containing codes achieves the classical and quantum Gilbert-Varshamov bounds, respectively. Furthermore, we show that the list decodability of symplectic dual-containing is upper bounded by the quantum Gilbert-Varshamov bound, namely, every symplectic dual-containing code with decoding radius $\delta$ and rate bigger than $ 1-H_q(\Gd)-\Gd\log_q(q+1)$ must have exponential list size.

A main technique is the powerful probabilistic  fact which says that
there is a limited correlation between the linear spaces and Hamming
balls. More precisely, it is unlikely that the intersection of a
linear subspace spanned by $t$ random vectors from a Hamming ball
has size more than $\Omega(t)$. This fact was used in \cite{Gru} and
is also an important ingredient in our proof.

Apart from the above fact, the counting idea on Euclidean (symplectic) self-orthogonal
linearly independent vectors and spaces by using solutions of
quadratic forms is of great important in computation of
probability.

\subsection{Organization}

The organization of this paper is as follows. We first review some
basic results on self-orthogonal codes and quadratic forms in
Subsections II.A and II.B. In Subsection II.3, we briefly discuss construction of random
Euclidean self-orthogonal codes based on quadratic forms. Subsection II.D presents list decoding and the Main Theorem I. Subsection II.E is fully
devoted to a proof of our Main Theorem I, i.e., Theorem \ref{thm:5}.
In Subsection II.F, we prove a lemma on  the number of certain
self-orthogonal spaces that is used in the proof of Theorem
\ref{thm:5}. Section III studies list decoding of symplectic dual-containing codes. We present a connection between decoding quantum stabilizer codes and symplectic dual-containing codes in Subsection III.C. Then we show that  list decodability of symplectic dual-containing is upper bounded by the quantum Gilbert-Varshamov bound in Subsection III.D.  Finally in Subsection III.E we prove our Main Theorem II which says that the list decodability of symplectic dual-containing codes achieves the quantum Gilbert-Varshamov bound.

\section{List decoding of Euclidean self-orthogonal codes}
\subsection{Euclidean self-orthogonal codes}
Let us quickly recall some basic concepts and results in coding
theory. As we focus on self-orthogonal codes  which are always
linear, we assume from now on that $q$ is a prime power and denote
by $\F_q$ the finite field of $q$ elements. A $q$-ary $[n,k]$-linear
code $C$ is a subspace of $\F_q^n$ with dimension $k$, where $n$
and $k$ are called the length and dimension of the code $C$,
respectively.  The information rate of the code $C$ is $R=k/n$ in
this case.

Two vectors $\bu$ and $\bv$ are said Euclidean orthogonal if $\langle
\bu,\bv\rangle=\sum_{i=1}^nu_iv_i=0$. A vector $\bu$ is said Euclidean
self-orthogonal if $\langle \bu,\bu\rangle=0$. The Euclidean dual code
$C^{\perp_E}$ of a linear code $C$ consists of all vectors in $\F_q^n$
that are orthogonal to every  codeword in $C$. A subset
$\{\bv_1,\dots,\bv_t\}$ of $\F_q^n$ is called Euclidean self-orthogonal if
$\langle \bv_i,\bv_j\rangle=0$ for all $1\le i,j\le t$.

A linear code $C$ is said Euclidean self-orthogonal if $C\subseteq C^{\perp_E}$.
It is easy to see that any Euclidean self-orthogonal code has dimension $k\leq
\frac{n}{2}$. Hence  a self-orthogonal code has information rate
$0\leq R\leq 1/2$.

\subsection{Quadratic forms}

An $n$-variate quadratic form over $\F_q$ is a zero polynomial or
homogeneous polynomial of degree $2$ in $n$ variables with
coefficients in $\F_q$, i.e., \[f(\bx)=f(x_1,\dots,
x_n)=\sum_{i,j=1}^n a_{ij}x_ix_j,\quad a_{ij}\in\F_q.\]

 A fundamental problem in the theory of quadratic form is how much one can simplify $f(\bx)$ by means of nonsingular linear transformation of indeterminates. Two quadratic forms $f(\bx)$ and $g(\bx)$ are said {\it equivalent} if there exists a nonsingular $n\times n$ matrix $A$ such that the quadratic form $f(\bx A)$ is equal to $g(\bx)$. It is easy to verify that this is indeed an equivalence relation. Furthermore, two equivalent quadratic forms have the same number of zeros. For a nonzero quadratic form $f(\bx)$,
 the smallest number $m$ for which $f(\bx)$ is not equivalent to a quadratic form in fewer than $m$ indeterminates  is called the {\it rank} of $f(\bx)$. The rank of the zero quadratic is defined to be $0$. If the rank of a nonzero quadratic form $f(\bx)$ is $n$, then $f(\bx)$ is called {\it non-degenerate}. If $\F_q$ has an odd characteristic, then a quadratic form $f(\bx)$ can be written as
$f(\bx)=\bx B\bx^T$ for a symmetric matrix $B$ of size $n$ over
$\F_q$. The rank of $B$ is in fact equal to the rank of $f(\bx)$.
The reader may refer to \cite[pages 278-289]{Nei} for the details
about quadratic forms.

For the purpose of this paper,  we are mainly interested in the
number of solutions of $f(\bx)=0$  for a quadratic form $f(\bx)$. We
combine several results in \cite[Section 6.2]{Nei} in the following
lemma.
\begin{lemma}\label{lem:2}
Let $f(\bx):=f(x_1,\dots, x_n)$ be a quadratic form defined over
$\F_q$ with rank $m$. Denote by $N(f(\bx)=0)$ the number of
solutions of $f(\bx)=0$ in $\F_q^n$. If $m=0$, then
$N(f(\bx)=0)=q^n$.   If $1\leq m\leq n$, then
\begin{eqnarray}\label{eq:1}
N(f(\bx)=0)=
\begin{cases}
q^{n-1} & \mbox{$m$ is odd};\\
q^{n-1}\pm(q-1)q^{n-\frac{m}{2}-1} & \mbox{$m$ is even}.\\
\end{cases}
\end{eqnarray}
\end{lemma}

\subsection{Construction of random Euclidean self-orthogonal  codes}
Unlike construction of a random linear code where one can choose a
random set of linearly independent vectors, construction of a random Euclidean
self-orthogonal code is not straightforward. In this part, we
briefly discuss  construction of  random Euclidean self-orthogonal codes
through quadratic forms.

Note that construction of a  random Euclidean self-orthogonal code is
equivalent to finding a linearly independent set
$\{\bv_1,\dots,\bv_k\}$ of random Euclidean self-orthogonal vectors.

We first choose a nonzero random solution
$\bv_1=(v_{11},\dots,v_{1n})$ of the quadratic equation
$x_1^2+\cdots+x_n^2=0$ (note that this equation has at least
$q^{n-2}$ solutions by Lemma \ref{lem:2}). Then $\bv_1$ is
self-orthogonal. Assume that we have already found a linearly
independent set $\{\bv_1,\dots,\bv_{k-1}\}$ of random Euclidean
self-orthogonal vectors. If we want to find a $k$th vector
$\bv_{k}=(v_{k1},\dots,v_{kn})$, then $(v_{k1},\dots,v_{kn})$ is a
solution  of the following  equation system
\begin{equation}\label{eq:02}\left\{
\begin{array}{l}
v_{11}x_1+\dots+v_{1n}x_n=0,\\
\vdots\\
v_{k-1,1}x_1+\dots+v_{k-1,n}x_n=0,\\
x_1^2+\dots+x_n^2=0.
\end{array}
\right.\end{equation}

Substituting the first $k-1$ equations of (\ref{eq:02}) into the
last equation of (\ref{eq:02}), we obtain a  quadratic  equation
$g(x_{i_1},\dots, x_{i_{n-k+1}})=0$ of $n-k+1$ variables. Thus, as
long as $N(g(x_{i_1},\dots, x_{i_{n-k+1}})=0)$ is bigger than the
cardinality of ${\rm span}\{\bv_1,\dots,\bv_{k-1}\}$, i.e,
$N(g(x_{i_1},\dots, x_{i_{n-k+1}})=0)>q^{k-1}$, we can randomly
choose a solution $\bv_k$ of (\ref{eq:02}) which is not contained in
${\rm span}\{\bv_1,\dots,\bv_{k-1}\}$ (note that the number of
solution of (\ref{eq:02}) is equal to $N(g(x_{i_1},\dots,
x_{i_{n-k+1}})=0))$. Hence, we obtain  a linearly independent set
$\{\bv_1,\dots,\bv_{k-1},\bv_k\}$ of random self-orthogonal vectors.

On the other hand, by Lemma \ref{lem:2}, the number
$N(g(x_{i_1},\dots, x_{i_{n-k+1}})=0)$ of solutions of $
g(x_{i_1},\dots, x_{i_{n-k+1}})=0$ is at least $q^{n-k-1}$. Thus, as
long as $q^{n-k-1}>q^{k-1}$, i.e, $k\le (n-1)/2$, we can proceed to
the next step to get a basis $\{\bv_1,\dots,\bv_{k-1},\bv_k\}$.

\subsection{List decoding random Euclidean self-orthogonal codes}

First of all, we assume that our channel has adversarial noise. In
other words, the channel can arbitrarily corrupt any subset of up to
a certain number of symbols of a codeword. Our goal is to correct
such errors and recover the original message/codeword efficiently.
An error-correcting code $C$ of block length $n$ over a finite
alphabet $\GS$ of size $q$ maps a set of messages into codewords in
$\GS^n$. The rate of the code $C$ is defined to be
$R:=R(C)=\frac{\log_q|C|}n$.

The formal definition of list decoding can be stated combinatorially
in the following way.
\begin{defn}\label{def:1.1}{\rm For  integers $q\ge 2$, $L\ge 1$ and a real $\Gd\in (0,1-1/q)$, a $q$-ary code $C$ of length $n$ over a code alphabet  $\Sigma$ of size $q$ is called $(\Gd,L)$-list-decodable if, for every point $\bx\in \Sigma^n$, there are at most $L$ codewords whose Hamming distance from $\bx$ is at most $\Gd n$.
}\end{defn}

Note that while considering $(\delta,
L)$-list-decodability, we always restrict the fraction $\delta<1-1/q
$ since decoding from a fraction of $1-1/q$  or more errors is
impossible except for the trivial codes. If we want a polynomial size list, the largest rate $R$ of the code
that one can hope is $1-H_q(\delta)$ \cite{Gru02,El,El91,ZP}.

The proof of our main theorem (Theorem \ref{thm:5}) combines an idea
used for random linear codes \cite{Gru} and  the counting result on
self-orthogonal linearly independent vectors and spaces by using
solutions of quadratic forms.

\begin{thm}\label{thm:5}(Main Theorem I)
For every prime power $q$ and a real $\delta \in (0,1-1/q)$
satisfying $1-H(\delta)\leq 1/2$, there exists a constant
$M_\delta$,  such that for  small $\varepsilon>0$ and all large
enough $n$, a random self-orthogonal code $C\subseteq \F_q^n$ of
rate $R=1-H(\delta)-\varepsilon$  is
$(\delta,\dfrac{M_\delta}{\varepsilon})$-list-decodable with
probability $1-q^{-n}$.
\end{thm}

The first step in the proof of Theorem \ref{thm:5}  is to reduce the
problem of the list-decodability of a random Euclidean self-orthogonal  code to
the problem of studying the weight distribution of certain random
linear code  containing a given Euclidean self-orthogonal code.

We quote a result from \cite{Gru} below where $B_n(\bx,\delta)$ denotes the Hamming ball with center $\bx\in
\F_q^n$ and radius $\delta n$.

\begin{lemma}\label{lem:3}
For every $\delta\in(0,1-1/q)$, there is a constant $M>1$ such that
for all $n$ and all $t=o(\sqrt{n})$, if $X_1,\dots, X_t$ are picked
independently and uniformly at random from $B_n(\bo,\delta)$, then

\[\Pr[|{\rm span}({X_1,\dots, X_{t}})\cap B_n(\bo,\delta)|\geq M\cdot t]\leq q^{-(6-o(1))n}.\]

\end{lemma}

This lemma shows that if we randomly pick $t$ vectors from the
Hamming ball $B_n(\bo,\delta)$, where $t$ is a constant depending on
the list size $L$, the probability that more than $\Omega(t)$
vectors in the span of these $t$ vectors lies in the ball
$B_n(\bo,\delta)$ is quite small. The detail of the proof for  this
lemma is given in \cite{Gru}. The key technique for proving this
lemma involves a Ramsey-theoretical lemma. Similar result for symplectic distance can be easily reduced to the case of Hamming distance by considering codes with alphabet size $q^2$ (see Section III).

The second step in our proof uses the following result on the
probability that a random linear code of dimension $k$ contains a
self-orthogonal subcode of dimension $k-1$ and a given set
$\{\bv_1,\dots,\bv_t\}\subseteq\F_q^n$ of linearly independent
vectors. Let $\mC_k^*$ denote the set of $q$-ary $[n,k]$-linear
codes in which every code contains a self-orthogonal subcode of
dimension $k-1$.

\begin{lemma}\label{lem:4}
For any linearly independent vectors $\bv_1,\dots,\bv_t$ in $\F_q^n$
with $t\leq k<n/2$, the probability that a random code $C^*$ from
$\mC_k^*$ contains $\{\bv_1,\dots,\bv_t\}$ is
\begin{eqnarray*}
&&\Pr_{C^*\in\mC_k^*}[\{\bv_1,\dots,\bv_t\}\subseteq C^*]\\
&& \leq
\begin{cases}
  q^{(k+t-n-1)t+2k-1} & \mbox{if  $q$ is even};\\
q^{(k+t-n-2)t+4k-2} & \mbox{if  $q$ is odd}.\\
 \end{cases}
 \end{eqnarray*}
 Hence, we always have \begin{equation}\label{eq:2} \Pr_{C^*\in\mC_k^*}[\{\bv_1,\dots,\bv_t\}\subseteq C^*]\le  q^{(k+t-n-2)t+4k-1}.\end{equation}
 \end{lemma}

We leave the proof of Lemma \ref{lem:4} to the coming Subsection F.

\subsection{Proof of Theorem \ref{thm:5}}

\begin{proof} Pick $M_{\Gd}=5M$, where $M$ is the constant in Lemma \ref{lem:3}.
Put $L=\lceil M_\delta/\epsilon\rceil$. Finally assume that $n$ is
large enough to satisfy (i) $n\ge L$; (ii) the term $o(1)$ in Lemma
\ref{lem:3} is at most $1$; (iii) $n\ge \frac
1{3(1-R)}(\log_q(2L)+L^2-L+3)$, i.e.,
\begin{equation}\label{eq:0a}
q^{-3(1-R)n}\times 2Lq^{L^2-L+3}\le 1.
\end{equation}

Let $C$ be a random self-orthogonal code with dimension $Rn$ in
$\F_q^n$. To show that $C$ is
$(\delta,\dfrac{M_\delta}{\varepsilon})$-list-decodable with high
probability, it is sufficient to show that with low probability that
$C$ is not  $(\delta,\dfrac{M_\delta}{\varepsilon})$-list-decodable,
i.e.,

\begin{equation}\label{eq:3}\Pr_{C\in\mC_{Rn}}[\exists \bx\in \F_q^n\ \mbox{such that}\ |B_n(\bx,\delta)\cap C|\geq L]<q^{- n},\end{equation}
where $\mC_k$ denotes the set of $q$-ary $[n,k]$-self-orthogonal
codes.

Thus, from now on we only need to prove that
\begin{equation}\label{eq:4}\Pr_{C\in\mC_{Rn},\bx\in\F_q^n}[|B_n(\bx,\delta)\cap C|\geq L]<q^{-n}\cdot q^{-(1-R)n}.\end{equation}
Note that the inequality (\ref{eq:4}) is derived from  (\ref{eq:3})
since, for every linear $C$ for which there is a ``bad" $\bx$ such
that $|B_n(\bx,\delta)\cap C|\geq L$, there are $q^{Rn}$ such ``bad"
$\bx$.

 Furthermore, the probability at the left side of (\ref{eq:4}) can be transformed into the following.
\begin{eqnarray*}
&&\Pr_{C\in\mC_{Rn},\bx\in\F_q^n}[|B_n(\bx,\delta)\cap C|\geq L]\\
&&=\Pr_{C\in\mC_{Rn},\bx\in\F_q^n}[|B_n(\bo,\delta)\cap (C+\bx)|\geq L]\\
  &&\leq \Pr_{C\in\mC_{Rn},{\bx}\in\F_q^n}[|B_n({\bo},\delta)\cap {\rm span}\{C, \bx\}|\geq L]\\
  &&\leq  \Pr_{C^*\in\mC_{Rn+1}^*}[|B_n(\bo,\delta)\cap C^*|\geq L],
 \end{eqnarray*}
where $C^*$ is a random $Rn+1$ dimensional subspace containing ${\rm
span}\{C, \bx\}$.

For any integer $t$ with $\log_q L\leq t\leq L$ (and hence $L\leq
q^t$), denote by $\mF_t$  the set of all tuples
$(\bv_1,\dots,\bv_t)\in B_n(\bo,\delta)^t$ such that
$\bv_1,\dots,\bv_t$ are linearly independent and $|{\rm
span}\{\bv_1,\dots, \bv_t\}\cap B_n(\bo,\delta)|\geq L$. Put
$\mF=\cup_{t=\lceil \log_q L\rceil}^L\mF_t$ and denote by $(\bv)$
and $\{\bv\}$  the tuple $(\bv_1,\dots,\bv_t)$ and the set
$\{\bv_1,\dots,\bv_t\}$, respectively.

We claim that if $|B_n(\bo,\delta)\cap C^*|\geq L$, there must
exist $(\bv)\in \mF$ such that $C^*\supseteq \{\bv\}$. Indeed, let
$\{\bu\}$ be a maximal linearly independent subset of
$B_n(\bo,\delta)\cap C^*$. If $|\{\bu\}|<L$, then we can simply take
$\{\bv\}=\{\bu\}$. Otherwise, we can take $\{\bv\}$ to be any subset
of $\{\bu\}$ of size $L$.
 Therefore, we have

\begin{eqnarray}
&&\Pr_{C^*\in\mC_{Rn+1}^*}[|B_n(\bo,\Gd)\cap C^*|\geq L] \\&&\leq \sum_{(\bv)\in \mF}\Pr_{C^*\in\mC_{Rn+1}^*}[C^*\supseteq \{\bv\}]\\
&&=\sum_{t=\lceil \log_q L\rceil} ^L\sum_{(\bv)\in \mF_t}\Pr_{C^*\in\mC_{Rn+1}^*}[C^*\supseteq \{\bv\}]\\
&&\leq \sum_{t=\lceil \log_q L\rceil} ^L|\mF_t|
q^{((Rn+1)+t-n-2)t+4(Rn+1)-1}\quad \mbox{by
(\ref{eq:2})}.\label{eq:7}
\end{eqnarray}

Thus, to have a good bound on our probability, we need to have  a
reasonably good upper bound for $|\mF_t|$. As in \cite{Gru}, we
divide the range of $t$ into two intervals.

(1) If $t< 5/\epsilon$, then
\[\dfrac{|\mF_t|}{|B_n(\bo,\delta)|^t}\leq \Pr[|{\rm span}({X_1,\dots, X_t})\cap B_n(\bo,\delta)|\geq L].\]

 Since $L\geq M\cdot t$, by Lemma \ref{lem:3} we have
 \[|\mF_t|\leq |B_n(\bo,\delta)|^t\cdot q^{-5n} \leq q^{ntH(\delta){-5n}}.\]

(2) If $t\geq 5/\epsilon$, then we  have $|\mF_t|\leq
|B_n(\bo,\delta)|^t\leq q^{ntH(\delta)}$ which is just a trivial
bound.

\hspace{0.5cm}

Finally, by substituting the value of $R=1-H(\delta)-\epsilon$ into
the inequality (\ref{eq:7}), we have
\[\begin{array}{lll}
&&\Pr_{C^*\in\mC_{Rn+1}^*}[|B_n(\bo,\delta)\cap C^*|\geq L]\\
&&\leq   \sum_{t=\lceil \log_q L\rceil}^{\lceil 5/\epsilon\rceil-1}q^{ntH(\delta)-5n}\cdot  q^{(-n+t+Rn-1)t+4Rn+3}\\
&&+\sum_{t=\lceil 5/\epsilon\rceil}^Lq^{ntH(\delta)}\cdot  q^{(-n+t+Rn-1)t+4Rn+3}\\
&&= q^{-5n+4Rn}\sum_{t=\lceil \log_q L\rceil}^{\lceil
5/\epsilon\rceil-1}q^{-\epsilon tn+t^2-t+3}\\ &&+q^{4Rn}\sum_{t=\lceil
5/\epsilon\rceil}^L
q^{-\epsilon tn+t^2-t+3}\\
&&\leq q^{-5n+4Rn}\cdot L\cdot q^{L^2-L+3}+q^{4Rn}\cdot L\cdot
q^{-5n+L^2-L+3}
\\
&&= q^{-n}\cdot q^{-(1-R)n}\times q^{-3(1-R)n}\times 2Lq^{L^2-L+3}\\
&&\le q^{-n}\cdot q^{-(1-R)n} \quad \mbox{by (\ref{eq:0a})}.
\end{array} \]
This completes the proof.
\end{proof}

\subsection{Proof of Lemma \ref{lem:4}}
Let us start with  a lemma that will be used in this subsection. Recall
that $\mC_k$ denotes the set of $q$-ary $[n,k]$ Euclidean self-orthogonal
codes, while $\mC_k^*$ denotes the set of $q$-ary $[n,k]$-linear
codes in which every code contains an Euclidean self-orthogonal subcode of
dimension $k-1$.

\begin{lemma}\label{lem:6} For any given linearly independent, self-orthogonal set $\{\bv_1,\dots,\bv_t\}$  with $t<k<n/2$ in a code $C^*\in\mC_k^*$,  one can find a self-orthogonal subcode $C'$ of $C^*$ with $\dim(C')=k-1$ such that $C'$ contains the set $\{\bv_1,\dots,\bv_t\}$.
\end{lemma}
\begin{proof} Let $V$ be the space spanned by $\{\bv_1,\dots,\bv_t\}$. Let $C$ be a   self-orthogonal subcode of $C^*$ of dimension $k-1$. If $V$ is a subspace of $C$, then we can simply take $C'=C$. If $t=k-1$, we can simply take $C'=V$. Now we assume that $V$ is not contained in $C$ and $t<k$. In this case, we must have $\dim(C\cap V)=t-1\le k-2$. Choose a vector $\bv$ from $V\setminus C$. Let $U\subseteq \F_q^n$ be the dual code of $\langle\bv\rangle$. Then $U$ has dimension $n-1$ and the intersection $U\cap C$ has dimension at least $k-2$. It is clear that $V\cap C$ is contained in  $U\cap C$ since $V\subseteq U$. Thus, $\bv$ is not contained in $U\cap C$. Furthermore, we can choose a subspace $W$ of $U\cap C$ such that $V\cap C\subseteq W$ and $\dim(W)=k-2$. Let $C'$ be the space spanned by $W$ and $\bv$. It is clear that $W$ is self-orthogonal since $W\subseteq C$. Moreover, $\bv$ is orthogonal to itself and every word in $W$ since $W\subseteq U$. As $\bv$ is not contained in $W$, $C'$ must have dimension $k-1$. This completes the proof.
\end{proof}

\subsection*{Case 1: $\F_q$ has even characteristic}

If $\F_q$ has even characteristic, then we have the following results from counting arguments.

\begin{lemma}\label{lem:7} For $k<n/2$,
the cardinality of $\mC_k^*$ is at least
\[\dfrac{(q^{n-k+1}-1)(q^{n-2k+3}-1)(q^{n-2k+5}-1)\cdots (q^{n-1}-1)}{(q^{k}-1)(q^{k-1}-1)\cdots (q-1)}.\]
\end{lemma}
\begin{proof}
Denote by $B_k^{(1)}$ and $B_k^{(2)}$ the cardinalities of $\mC_k$
and  $\mC^*_k\setminus \mC_k$, respectively. Then
$|\mC^*_k|=B_k^{(1)}+B_k^{(2)}.$

Let us consider  $B_k^{(1)}$ first. First of all, a vector
$\bx=(x_1,\dots,x_n)$ is self-orthogonal if and only if
$x_1^2+\cdots+x_n^2=0$. This quadratic form is equivalent to
$x_1^2=0$ and hence by Lemma \ref{lem:2} it has $q^{n-1}$ solutions.

 For a code $C_{i-1}$ in $\mC_{i-1}$, we can span  $C_{i-1}$ into a self-orthogonal code $C_i$ in $\mC_i$ by adding one self-orthogonal vector in $C_{i-1}^{\perp_E}\diagdown C_{i-1}$.    Hence, we have $q^{n-i}-q^{i-1}$ choices of such a vector. On the other hand, there are $(q^k-q^{k-1})(q^k-q^{k-2})\cdots (q^k-1)$ choices of $k$-dimensional  basis generating the same code of dimension $k$. Therefore,
\[B_k^{(1)}\ge \dfrac{(q^{n-2k+1}-1)(q^{n-2k+3}-1)\cdots (q^{n-1}-1)}{(q^k-1)(q^{k-1}-1)\cdots (q-1)}.\]

The computation of $B_k^{(2)}$ is a bit different  from that of
$B_k^{(1)}$ as codes in   $\mC^*_k\setminus \mC_k$ are not
self-orthogonal. We first choose a linearly independent,
self-orthogonal set of size $k-1$. One can then span this set into a
code in $\mC^*_k\setminus \mC_k$ by adding a  vector in
$\F_q^n\diagdown C_{k-1}^{\perp_E}$. Thus, we obtain a recursive
formula and get the following inequality

\[B_k^{(2)}\geq \dfrac{(q^{n-k+1}-q^{n-2k+1})\prod_{i=1}^{k-1}(q^{n-2i+1}-1)}{(q^{k}-1)(q^{k-1}-1)\cdots (q-1)}. \]

The desired result follows from adding $B_k^{(1)}$ with $B_k^{(2)}$.
\end{proof}

\begin{lemma}\label{lem:8}
For given $t$ ($t<k<n/2$) linearly independent vectors $\bv_1,\dots
,\bv_t$ in $\F_q^n$, the number of linear codes $C^*\in\mC_k^*$ such
that $C^*\supseteq \{\bv_1,\dots,\bv_t\}$ is at most
\[\dfrac{(q^{n-2k+3}-1)(q^{n-2k+5}-1)\cdots (q^{n-2t-1}-1)(q^n-q^{k-1})}{(q^{k-t-1}-1)(q^{k-t-2}-1)\cdots (q-1)}
.\]

\end{lemma}
\begin{proof} Denote by $A_k$ the number  of linear codes $C^*\in\mC_k^*$ such that $C^*\supseteq \{\bv_1,\dots,\bv_t\}$.
 Denote by $D$ a maximal self-orthogonal code in ${\rm span}\{\bv_1, \dots,\bv_t\}$. Let  $A_k^{(1)}$ denote  the number of self-orthogonal codes $C^*$ in  $\mC_k^*$  such that $\{\bv_1,\dots,\bv_t\}\subseteq C^*$;  and let $A_k^{(2)}$ denote the number of $C^*\in\mC_k^*$ such that $C^*$ is not self-orthogonal  and $\{\bv_1,\dots,\bv_t\}\subseteq C^*$.

{\it Case 1}: If ${\dim} (D)=t$, then $\{\bv_1, \dots,\bv_t\}$ is a
self-orthogonal set. By Lemma \ref{lem:6}, we can span $\{\bv_1,
\dots,\bv_t\}$ into a code in $\mC^*_k$.

The counting idea is similar to that in the proof of Lemma
\ref{lem:7} except for that we  first fix $t$ linearly independent
vectors $\{\bv_1, \dots,\bv_t\}$ and then span them into a larger
code. Thus, we have
\[A_k^{(1)}\leq\dfrac{\prod_{i=t-1}^{k-1}(q^{n-2i+1}-1)(q^{n-k}-q^{k-1})}{(q^{k-t}-1)(q^{k-t-1}-1)\cdots (q^2-1)(q^{k}-q^{k-1})}.
\]
Similarly, we have
 \[A_k^{(2)}\leq \dfrac{\prod_{i=t-1}^{k-1}(q^{n-2i+1}-1)(q^{n}-q^{n-k})}{(q^{k-t-1}-1)(q^{k-t-2}-1)\cdots (q-1)}
.\] The desired result follows from adding $A_k^{(1)}$ with
$A_k^{(2)}$.

{\it Case 2}: If ${\dim} (D)=t-1$, then we choose a suitable basis
$\bu_1,\dots,\bu_t$ for ${\rm span}\{\bv_1, \dots,\bv_t\}$ such that
$\bu_1,\dots,\bu_{t-1}\in D$. In this case, by Lemma \ref{lem:6} we
can get a code $C'$ of dimension $k-1$ that contains $D$, and then a
code $C^*:={\rm span}\{C',\bu_t\}$. Hence,
\[A_k\leq \dfrac{(q^{n-2k+3}-1)(q^{n-2k+5}-1)\cdots (q^{n-2t+1}-1)}{(q^{k-t}-1)(q^{k-t-1}-1)\cdots (q-1)}
.\]

{\it Case 3}: If ${\dim} (D)\leq t-2$, then in this case it is
impossible  to find a  code in $\mC^*_k$ containing  $\{\bv_1,
\dots,\bv_t\}$. In other words, $A_k=0$.

This completes the proof.
\end{proof}

\subsection*{Case 2: $\F_q$ has odd characteristic }
The counting technique for odd $q$ is analogous with that of even
$q$. The only difference here is the number of self-orthogonal
vectors.

Note that a vector $\bx=(x_1,\dots,x_n)\in\F_q^n$ is self-orthogonal
if and only if
\begin{equation}\label{eq:01}
x_1^2+\cdots+x_n^2=0.
\end{equation}
In the case where $q$ is even, the quadratic form (\ref{eq:01}) has
rank $1$. Hence, by Lemma \ref{lem:2} it has $q^{n-1}$ solutions.
However, in the case where $q$ is odd, the quadratic form
(\ref{eq:01}) has rank $n$ and hence by Lemma \ref{lem:2} the number
of its solutions is between $q^{n-1}-(q-1)q^{\frac n2-1}$ and
$q^{n-1}+(q-1)q^{\frac n2-1}$. Therefore, the corresponding results
of Lemmas \ref{lem:7} and \ref{lem:8} are slightly different in the
case of odd characteristic. We state the results below without
proofs.

\begin{lemma}\label{lem:9} For $k<n/2$,
 the cardinality of $\mC_k^*$ is at least
\[\dfrac{(q^{n-k+1}-1)(q^{n-2k+2}-1)(q^{n-2k+4}-1)\cdots (q^{n-2}-1)}{(q^{k}-1)(q^{k-1}-1)\cdots (q-1)}.\]
\end{lemma}

\begin{lemma}\label{lem:10}
For given $t$ ($t<k<n/2$) linearly independent vectors $\bv_1,\dots
,\bv_t$ over $\F_q^n$, the number of linear codes $C^*\in\mC_k^*$
such that $C^*\supseteq \{\bv_1,\dots,\bv_t\}$ is at most
\[\dfrac{\prod_{i=t-1}^{k-1}(2q^{n-2i+1}-1)(q^n+q^{n-2k+1})}{(q^{k-t-1}-1)(q^{k-t-2}-1)\cdots (q-1)}
.\]
\end{lemma}

{\it\bf Proof of Lemma \ref{lem:4}}: For even $q$, by Lemmas
\ref{lem:7} and  \ref{lem:8}, we have
 \begin{eqnarray*}&&\Pr_{C^*\in\mC^*_k}[\{\bv_1,\dots,\bv_t\}\subseteq C^*]\\
 &&=\frac{|\{C^*\in\mC_k:\; C^*\supseteq \{\bv_1,\dots,\bv_t\}\}|}{|\mC_k^*|}\\
 &&\leq q^{2k-t-1}\left(\dfrac{q^{k-t+1}}{q^{n-2t+1}}\right)^{t}\leq q^{(k-n+t-1)t+2k-1}. \end{eqnarray*}

For odd $q$, by Lemmas \ref{lem:9} and  \ref{lem:10}, we have
 \begin{eqnarray*}
&&\Pr_{C^*\in\mC^*_k}[\{\bv_1,\dots,\bv_t\}\subseteq C^*]\\ &&=\frac{|\{C^*\in\mC_k:\; C^*\supseteq \{\bv_1,\dots,\bv_t\}\}|}{|\mC_k^*|}\\
&&\leq \left(\dfrac{2q^{n-2k+3}-1}{q^{n-2k+2}-1}\right)^{k-t-1}\cdot \left(\dfrac{q^{k-t}-1}{q^{n-2t}-1}\right)^{t} \\
&&\cdot \dfrac{(q^{k}-1)(q^n+q^{n-2k+1})}{q^{n-k+1}-1} \\
&&\leq (3q)^{k-t-1}\left(\dfrac{q^{k-t}}{q^{n-2t}}\right)^t q^{2k}\\
&&\leq  q^{(-n+t+k-2)t+4k-2}.
\end{eqnarray*}
This completes the proof.
\epf

\section{List-decoding of symplectic self-orthogonal codes}
\subsection{Symplectic self-orthogonal codes}
To define symplectic inner product, we have to consider a $q$-ary $[2n,k]$-linear
code $C$ in $\F_q^{2n}$. Two vectors $(\bu_1|\bv_1)$ and $(\bu_2|\bv_2)$ are said symplectic  orthogonal if $\langle
\bu_1,\bv_2\rangle-\langle
\bu_2,\bv_1\rangle=0$. Note that every vector $(\bu|\bv)$ is symplectic
self-orthogonal. The dual code
$C^{\perp_S}$ of a linear code $C$ consists of all vectors in $\F_q^{2n}$
that are orthogonal to every  codeword in $C$. A subset
$\{(\bu_1|\bv_1),\dots,(\bu_t|\bv_t)\}$ of $\F_q^{2n}$ is called symplectic self-orthogonal if
the symplectic inner product of $(\bu_i|\bv_i)$ and $(\bu_j|\bv_j)$ are $0$ for all $1\le i,j\le t$.

A linear code $C$ is said symplectic self-orthogonal if $C\subseteq C^{\perp_S}$. It is well known that a $q$-ary $[2n,k]$-symplectic self-orthogonal code gives a $q$-ary $[[n,n-k]]$-quantum code \cite{CRSS98}. Thus, we define the rate of $C$ in terms of the associate quantum code, i.e., $R:=(n-k)/n$.

Finally, let us  define symplectic weight and distance. For a vector $(\bu|\bv)=(u_1,\dots,u_n|v_1,\dots,v_n)\in\F_q^{2n}$, the symplectic weight is defined to be $\wt_S(\bu|\bv)=|\{1\le i\le n:\; (u_i,v_i)\not=(0,0)\}|.$ The symplectic distance of two vectors  $(\bu_1|\bv_2)$ and $(\bu_2|\bv_2)$ is defined to be $\wt_S(\bu_1-\bu_2|\bv_1-\bv_2)$.

\subsection{Construction of symplectic self-orthogonal codes}

Compared with construction of  random Euclidean self-orthogonal codes, construction of random  symplectic self-orthogonal codes is much easier. This is because every vector in $\F_q^{2n}$ is self-orthogonal under the symplectic inner product. Again construction of a  random  symplectic self-orthogonal code is
equivalent to finding a linearly independent set
$\{(\bu_1|\bv_1),\dots,(\bu_t|\bv_t)\}$  of random  symplectic self-orthogonal vectors. We first choose a nonzero random vector
$(\bu_1|\bv_1)=(u_{11},\dots,u_{1n}|v_{11},\dots,v_{1n})$. Assume that we have already found a linearly
independent set $\{(\bu_1|\bv_1),\dots,(\bu_{k-1}|\bv_{k-1})\}$ of random symplectic
self-orthogonal vectors. If we want to find a $k$th vector
$(\bu_k|\bv_k)=(u_{k1},\dots,u_{kn}|v_{k1},\dots,v_{kn})$, then $(u_{k1},\dots,u_{kn},v_{k1},\dots,v_{kn})$ is a
solution  of the following  equation system
\begin{equation}\label{eq:12}\left\{
\begin{array}{l}
v_{11}x_1+\dots+v_{1n}x_n-(u_{11}y_1+\dots+u_{1n}y_n)=0,\\
\vdots\\
v_{k-1,1}x_1+\dots+v_{k-1,n}x_n-\\
(u_{k-1,1}y_1+\dots+u_{k-1,n}y_n)=0.\\
\end{array}
\right.\end{equation}

\subsection{Connection between  decoding of quantum stabilizer codes and decoding of symplectic self-orthogonal codes}
To simplify our presentation in this subsection, we consider only binary quantum stabilizer codes. Let us briefly describe the background on quantum stabilizer codes and their decoding. The reader may refer to \cite{CRSS98, KL,NC} for the details on decoding of quantum stabilizer codes.

The state space of one qubit is actually a $2$-dimensional complex space with a basis $\{|0\rangle, |1\rangle\}$. We can simply denote this state space of one qubit by $\CC^2$.  Let $\mG_1=\{\pm I, \pm iI, \pm X, \pm iX, \pm Y, \pm i Y, \pm Z, \pm i Z\} $ be the Pauli group acting on $\CC^2$, where $i$ is the imaginary unit, $I$ is the $2\times 2$ identity matrix and
\[X=\left(\begin{array}{cc}0&1\\ 1& 0\end{array}\right),\quad Z=\left(\begin{array}{cc}1&0\\ 0& -1\end{array}\right),\quad Y=iXZ.\]

The tensor product $(\CC^2)^{\otimes n}$  is called the state space of $n$ qubits.  Let $\mG_n$ denote the Pauli group acting on $(\CC^2)^{\otimes n}$ , i.e.,
\[\mG_n=\{i^m\Gs_1\otimes\Gs_2\otimes\cdots\otimes\Gs_n:\; m\in\{0,1,2,3\}, \ \Gs_j\in\{I,X,Y,Z\}\},\]
where the action of an element of $\mG_n$  on a state of $n$ qubits is through the componentwise action of $\Gs_i$ on $\CC^2$.

Quantum stabilizer codes are defined in the following manner. Let $\mS$ be a subgroup of $\mG_n$ such that $-I\otimes I\otimes\cdots\otimes I\not\in \mS$. Then $\mS$ is a $2$-elementary abelian group. Assume that the $2$-rank of $\mS$ is $k$  for some $k\in[0, n]$ and $\mS$ is generated by $\{g_1,g_2,\dots,g_{k}\}$. The subgroup $\mS$ has a fixed subspace $Q_{\mS}$ of $(\CC^2)^{\otimes n}$ defined by
\[Q_{\mS}=\{\bv\in (\CC^2)^{\otimes n}:\; g(\bv)=\bv\ \mbox{for all $g\in \mS$}\}.\]
The subspace $Q_{\mS}$ is called an $[[n,n-k]]$-quantum stabilizer code and it has dimension $2^{n-k}$.

To connect the quantum stabilizer code $Q_{\mS}$ with a classical linear code, we define a group epimorphism $\psi:\; \mG_n\rightarrow \F_2^{2n}$ given by
\[\psi(i^m\Gs_1\otimes\Gs_2\otimes\cdots\otimes\Gs_n)=(x_1,x_2\dots,x_n|z_1,z_2\dots,z_n)=(\bx|\bz),\]
where $x_j,z_j$ are elements of $\F_2$ that are determined as below
\[
\begin{array}{c|cccc} \Gs_j &I&X&Y&Z\\ \hline
x_j&0&1&1&0\\
z_j&0&0&1&1
\end{array}\]
Furthermore, we define a $2n\times 2n$ matrix over $\F_2$
\[\Lambda=\left(\begin{array}{cc}O&I\\ I &O\end{array}\right),\]
where $O$ is the $n\times n$ zero matrix and $I$ is the $n\times n$ identity matrix. Then it is easy to see that, for two elements $g,h\in \mG_n$, $gh=hg$ if and only if $\psi(g)\Lambda\psi(h)^T=0$, i.e., $\psi(g)$ and $\psi(h)$ are symplectic self-orthogonal. Through the $k$ generators $\{g_1,g_2,\dots,g_{k}\}$, we define an $k\times 2n$ matrix over $\F_2$
\[H=\left(\begin{array}{c}\psi(g_1)\\ \cdot\\ \cdot\\ \cdot\\ \psi(g_k)\end{array} \right).\]
It is easy to see that $H$ has rank $k$. Since $\mS$ is abelian, we have $H\Lambda H^T={\bf 0}$. Thus, the binary code $C$ with $H$ as a generator matrix is symplectic self-orthogonal.

Now we briefly review decoding of quantum stabilizer codes. Consider an $[[n,n-k]]$-quantum  stabilizer code $Q_{\mS}$ as defined above. Assume that a state of $n-k$ quibits is encoded into a coded state $|\alpha\rangle$ of $n$ qubits. Let $\rho=|\alpha\rangle\langle\Ga|$ be the channel input and let $E\rho E^{\dag}$ be the channel output with error $E\in\mG_n$, where $E^{\dag}$ denotes the Hermitian conjugation of $E$ . By computing the syndrome measurements of the received state, one can determine the binary syndrome $\bs$ which is equal to $\psi(E)\Lambda H^T$ (see \cite{KL}). To decode, i.e., recover the channel input $\rho$, it is sufficient to determine  the error $E$. On the other hand, finding $E$ can be reduced to finding $\psi(E)$  (note that the scalar $i^m$ does not affect error). Thus, we turn the problem of decoding quantum stabilizer codes into decoding of $C^{\perp_S}$ (to see this, we notice that $H$ is a parity-check matrix of  $C^{\perp_S}$).
Assume that $E$ has at most $t$ errors, i.e., in the representation $E=i^m\Gs_1\otimes\Gs_2\otimes\cdots\otimes\Gs_n$, there are at most $t$ indices $j$ such that $\Gs_j\not=I$. Thus, the corresponding binary vector $\psi(E)$ has symplectic weight at most $t$. This implies that we have to find an error $\be=\psi(E)\in\F_2^{2n}$ such that $\wt_S(\be)\le t$ and $\be \Lambda H^T=\bs$. This is exactly the decoding problem of classical codes. To list decode $Q_{\mS}$, we can find the list of all vectors $\be\in\F_2^{2n}$ such that $\wt_S(\be)\le t$ and $\be \Lambda H^T=\bs$. In other words, if $\bx_0$ is a solution of $\bx \Lambda H^T=\bs$, then we have to find all codewords $\bc\in C^{\perp_S}$ such that $\wt_S(\bc-\bx_0)\le t$.

\subsection{Upper bound on list decodability of symplectic self-orthogonal codes}
Recall that the list decodability of classical block codes is upper bounded by  the classical Gilbert-Varshamov bound (\cite{Gru01}).  In this subsection, we show a similar result for symplectic self-orthogonal codes, namely, the list decodability of symplectic self-orthogonal codes is upper bounded by the quantum Gilbert-Varshamov bound.

First, we have to give a formal definition of list decoding for a symplectic dual-containing code.

\begin{defn}\label{def:1.2}{\rm For  a prime $q\ge 2$, an integer $L\ge 1$ and a real $\Gd\in (0,1/2)$, a $q$-ary symplectic self-orthogonal code $C$ of length $2n$ over a code alphabet  $\F_q$ is called $(\Gd,L)$-list-decodable if, for every point $\bx\in \F_q^{2n}$, there are at most $L$ codewords in $C^{\perp_S}$ whose symplectic distance from $\bx$ is at most $\Gd n$.
}\end{defn}

Note that list decoding of $C$ is in fact list decoding of its symplectic dual $C^{\perp_S}$.
\begin{thm}\label{thm:upper}
For every prime power $q$ and a real $\delta \in (0,1/2)$, a $q$-ary symplectic self-orthogonal code $C$ of length $2n$, decoding radius $\delta$ and
rate $R>1-H_q(\delta)-\Gd\log_q(q+1)+o(1)$ must have an exponential list size in $n$.
\end{thm}
\begin{proof}
Let $k$ be the dimension of $C$. Then the rate of $C$ is $R=(n-k)/n$.  Pick up a random word $\bx\in\F_q^{2n}$ and consider the random variable $X:=|B_{2n}^{S}(\bx,\delta)\cap C^{\perp_S}|$,     where $B_{2n}^{S}(\bx,\delta)$ is the symplectic ball of radius $\Gd n$, i.e., $B_{2n}^S(\bx,\delta)$ consists of all vectors of $\F_q^{2n}$ that have symplectic distance at most $\Gd n$ from $\bx$. The expected value of $X$ is clearly $|C^{\perp_S}|\cdot |B_{2n}^{S}(\bo,\delta)|/q^{2n}$ which is at least
\begin{eqnarray*}
&&q^{2n-k}\times q^{nt(H_q(\Gd)+\Gd\log_q(q+1))}\times q^{-2n}\\ &&=q^{n(R-(1-H_q(\Gd)-\Gd\log_q(q+1)))}=\Omega(\exp(n)).
\end{eqnarray*}
This completes the proof.
\end{proof}

\subsection{List decoding random  symplectic self-orthogonal codes}

Now we state the list decodability of random  symplectic self-orthogonal codes below.
\begin{thm}(Main Theorem II)\label{thm:2}
For every prime power $q$ and a real $\delta \in (0,1/2)$, there exists a constant
$M_\delta$,  such that for every small $\varepsilon>0$ and all large
enough $n$, a $q$-ary random symplectic self-orthogonal code $C$ of length $2n$ and
rate $R=1-H_q(\delta)-\Gd\log_q(q+1)-\varepsilon$  is
$(\delta,\dfrac{M_\delta}{\varepsilon})$-list-decodable with
probability $1-q^{-n}$.
\end{thm}

The proof of Theorem \ref{thm:2} is exactly similar to the one of Theorem \ref{thm:5} except for the different counting of symplectic self-orthogonal codes. For preparation, we give two lemmas that are needed for the proof of Theorem \ref{thm:2}.

 By considering Hamming ball over alphabet size $q^2$, we get a similar result as in Lemma \ref{lem:3}.

\begin{lemma}\label{lem:a1}
For every $\delta\in(0,1-1/q)$, there is a constant $M>1$ such that
for all $n$ and all $t=o(\sqrt{n})$, if $X_1,\dots, X_t$ are picked
independently and uniformly at random from $B_{2n}^{S}(\bo,\delta)$, then
\[\Pr[|{\rm span}({X_1,\dots, X_{t}})\cap B_{2n}^S(\bo,\delta)|\geq M\cdot t]\leq q^{-2(6-o(1))n}.\]
\end{lemma}

Next we prove a result on probability for a symplectic dual-containing code containing  a given  set $\{\bv_1,\dots,\bv_t\}$ of linearly independent vectors in $\F_q^{2n}$.

Let $\mS_k$ denote the set of $k$-dimensional symplectic self-orthogonal codes in $\F_q^{2n}$.

\begin{lemma}\label{lem:a2}
For any linearly independent vectors $\bv_1,\dots,\bv_t$ in $\F_q^{2n}$, the probability that a random code $C$ from
$\mS_k$ with $C^{\perp_S}$ containing $\{\bv_1,\dots,\bv_t\}$ satisfies
\begin{equation}\label{eq:a1} \Pr_{C\in\mS_k}[\{\bv_1,\dots,\bv_t\}\subseteq C^{\perp_S}]\le  q^{-kt}.\end{equation}
 \end{lemma}
\begin{proof} Let us first compute the size of $\mS_k$.
Note that every element of $\F_q^{2n}$ is symplectic self-orthogonal. Thus, every $k$-dimensional self-orthogonal code $C_k\in\mS_k$ is spanned from a $k-1$-dimensional self-orthogonal code $C_{k-1}\in\mS_{k-1}$ by adding a vector in $C_{k-1}^{\perp_S}\backslash C_{k-1}$. Given the fact that, for two vectors $\bu,\bv\not\in C_{k-1}$, $\langle \bu, C_{k-1}\rangle=\langle \bv, C_{k-1}\rangle$ if and only if $\bu-\Gl \bv\in C_{k-1}$ for some nonzero $\Gl\in \F_q$,
 we know that the number  of symplectic self-orthogonal  codes $C_k$ containing a fixed symplectic self-orthogonal $C_{k-1}$ is $(q^{2n-2k+2}-1)/(q-1)$.  On the other hand, every $k$-dimensional symplectic self-orthogonal code contains exactly $(q^k-1)/(q-1)$ symplectic self-orthogonal spaces of dimension $k-1$. This gives the recursive formula $|\mS_k|(q^k-1)/(q-1)=|\mS_{k-1}|(q^{2n-2k+2}-1)/(q-1)$.
From this recursive formula, we obtain
\begin{equation}\label{eq:a2} |\mS_k|=\frac{(q^{2n-2k+2}-1)(q^{2n-2k+4}-1)\cdots(q^{2n}-1)}{(q^{k}-1)(q^{k-1}-1)\cdots(q-1)}.\end{equation}
Let $V$ be the linear span of  $\bv_1,\dots,\bv_t$. Then $C^{\perp_S}$ contains $\bv_1,\dots,\bv_t$ in $\F_q^{2n}$ if and only if $C$ is a subspace of $V^{\perp_S}$. Thus, the number of symplectic self-orthogonal codes $C$ with $C^{\perp_S}$ containing $\{\bv_1,\dots,\bv_t\}$ is in fact the number of symplectic self-orthogonal codes in $V^{\perp_S}$. Since $\dim V^{\perp_S}=2n-t$, by (\ref{eq:a2}) this number is
 at most
\begin{equation}\label{eq:a3} \frac{(q^{2n-t-2k+2}-1)(q^{2n-t-2k+4}-1)\cdots(q^{2n-t}-1)}{(q^{k}-1)(q^{k-1}-1)\cdots(q-1)}.\end{equation}
Dividing (\ref{eq:a3}) by (\ref{eq:a2}) gives the desired result.
\end{proof}

{\it\bf Proof of Theorem \ref{thm:2}}: Pick $M_{\Gd}=4M$, where $M$ is the constant in Lemma \ref{lem:a1}.
Put $L=\lceil M_\delta/\epsilon\rceil$. Assume that $n$ is sufficiently
large.

Let $C$ be a random symplectic self-orthogonal code with rate $R$, i.e., dimension $k$ of $C$ satisfies $k=(1-R)n$ in
$\F_q^{2n}$. To show that $C^{\perp_S}$ is
$(\delta,\dfrac{M_\delta}{\varepsilon})$-list-decodable with high
probability, it is sufficient to show that with low probability that
$C^{\perp_S}$ is not  $(\delta,\dfrac{M_\delta}{\varepsilon})$-list-decodable,
i.e.,

\begin{equation}\label{eq:a5}\Pr_{C\in\mS_{k}}[\exists \bx\in \F_q^{2n}\ \mbox{such that}\ |B_{2n}^S(\bx,\delta)\cap C^{\perp_S}|\geq L]<q^{- n}.\end{equation}

Thus, from now on we only need to prove that
\begin{equation}\label{eq:a6}\Pr_{C\in\mS_{k},\bx\in\F_q^{2n}}[|B_{2n}^S(\bx,\delta)\cap C^{\perp_S}|\geq L]<q^{-n}\cdot q^{k-2n}.\end{equation}

 Furthermore, the probability at the left side of (\ref{eq:a6}) can be transformed into the following.
\begin{eqnarray}
&&\Pr_{C\in\mS_{k},\bx\in\F_q^n}[|B_{2n}^S(\bx,\delta)\cap C^{\perp_S}|\geq L]\\&& \leq \Pr_{D\in\mS_{k-1}}[|B_{2n}^S(\bo,\delta)\cap D^{\perp_S}|\geq L]\\
&&=\sum_{t=\lceil \log_q L\rceil} ^L\sum_{(\bv)\in \mF_t}\Pr_{D\in\mS_{k-1}}[D^{\perp_S}\supseteq \{\bv\}]\\
&&\leq \sum_{t=\lceil \log_q L\rceil} ^L|\mF_t|\cdot
q^{-kt},\quad \mbox{by
(\ref{eq:a1})}\label{eq:a10}
\end{eqnarray}
where $D^{\perp_S}$ is a random $2n-k+1$ dimensional subspace containing ${\rm
span}\{C^{\perp_S}, \bx\}$ and $\mF_t$ is defined in the proof of Theorem \ref{thm:5}. Note that we use  the fact that ${\rm
span}\{C^{\perp_S}, \bx\}$ is symplectic dual-containing whenever $C^{\perp_S}$ is.

(1) If $t< 4/\epsilon$, then
\[\dfrac{|\mF_t|}{|B_{2n}^S(\bo,\delta)|^t}\leq \Pr[|{\rm span}({X_1,\dots, X_t})\cap B_{2n}^S(\bo,\delta)|\geq L].\]

 Since $L\geq M\cdot t$, by Lemma \ref{lem:a1} we have
 \begin{eqnarray*}
 &&|\mF_t|\leq |B_{2n}^S(\bo,\delta)|^t\cdot q^{-10n} \leq q^{2ntH_{q^2}(\delta){-10n}}\\
 &&=q^{nt(H_q(\Gd)+\Gd\log_q(q+1))-10n}.
 \end{eqnarray*}

(2) If $t\geq 4/\epsilon$, then we  have $|\mF_t|\leq
|B_{2n}^S(\bo,\delta)|^t=q^{nt(H_q(\Gd)+\Gd\log_q(q+1))}$ which is just a trivial
bound.

\hspace{0.5cm}

Finally,  substituting the value of $k=(1-R)n$ and $R=1-H_q(\delta)-\Gd\log_q(q+1)-\epsilon$ into
the inequality (\ref{eq:a10}), we get
\[\begin{array}{lll}
&&\Pr_{C\in\mS_{k},\bx\in\F_q^n}[|B_{2n}^S(\bx,\delta)\cap C^{\perp_S}|\geq L]\\
&&\leq
  \sum_{t=\lceil \log_q L\rceil}^{\lceil 4/\epsilon\rceil-1}q^{nt(H_q(\Gd)+\Gd\log_q(q+1))-10n}\cdot q^{-kt}\\
&&+\sum_{t=\lceil 4/\epsilon\rceil}^Lq^{nt(H_q(\Gd)+\Gd\log_q(q+1))}\cdot q^{-kt}\\

&&\leq\sum_{t=\lceil \log_q L\rceil}^{\lceil 4/\epsilon\rceil-1} q^{-\epsilon nt-10n}+\sum_{t=\lceil 4/\epsilon\rceil}^Lq^{-\epsilon nt}
\\

&&\le q^{-n}\cdot q^{k-2n}.
\end{array} \]
This completes the proof.

\section*{Acknowledgments}
The authors are grateful to the anonymous referees and Professor Dr. Alexei Ashikhmin for their invaluable and constructive
comments and suggestions which have greatly improved the structure and presentation of this paper and make this paper more readable.

{\bf Lingfei Jin} received her B.A. degree in mathematics from Hefei University of
Technology, China in 2009. In 2013, she received her Ph.D. degree from Nanyang Technological University, Singapore. Currently, she is an associate Professor at the School of Computer Science, Fudan University, China.
Her research interests include quantum information, cryptography and
coding theory.\\ \\
{\bf Chaoping Xing} received his Ph.D. degree in 1990 from University of Science
and Technology of China. From 1990 to 1993 he was a lecturer and associate
professor in the same university. He joined University of Essen, Germany as an
Alexander von Humboldt fellow from 1993 to 1995. After this he spent most
time in Institute of Information Processing, Austrian Academy of Sciences until
1998. From March of 1998 to November of 2007, he was working in National
University of Singapore. Since December of 2007, he has been with Nanyang
Technological University and currently is a full Professor. Dr. Xing has been
working on the areas of algebraic curves over finite fields, coding theory, cryptography
and quasi-Monte Carlo methods, etc.\\ \\
{\bf Xiande Zhang} received the Ph.D. degree in mathematics from Zhejiang University, Hangzhou, Zhejiang, P. R. China in 2009. After that, she held postdoctoral positions in Nanyang Technological University and Monash University. Currently, she is a Research Fellow at the Division of Mathematical Sciences, School of Physical and Mathematical Sciences, Nanyang Technological University, Singapore.  Her research interests include combinatorial design theory, coding theory, cryptography, and their interactions.

\end{document}